\documentclass[a4paper,UKenglish]{lipics}
\usepackage{microtype}%if unwanted, comment out or use option "draft"
\title{Asymmetry of the Kolmogorov complexity of online predicting odd and even bits}
\titlerunning{Asymmetry of online Kolmogorov complexity} 

\author{Bruno Bauwens}
\affil{Universit\'e de Lorraine, LORIA\\ %This research was supported by NAFIT ANR-08-EMER-008-01 project.
  615, Rue du Jardin Botanique, France\\
    \texttt{Brbauwens at gmail dot com}
  }
\authorrunning{B. Bauwens} %mandatory. First: Use abbreviated first/middle names. Second (only in severe cases): Use first author plus 'et. al.'

\Copyright{B. Bauwens}%mandatory. LIPIcs license is "CC-BY";  http://creativecommons.org/licenses/by/3.0/

\subjclass{E.4} 
\keywords{(On-line) Kolmogorov complexity, (On-line) Algorithmic Probability, Philosophy of
 Causality, Information Transfer}

\serieslogo{}%please provide filename (without suffix)
\volumeinfo%(easychair interface)
  {Billy Editor, Bill Editors}% editors
  {2}% number of editors: 1, 2, ....
  {Conference title on which this volume is based on}% event
  {1}% volume
  {1}% issue
  {1}% starting page number
\EventShortName{}

\usepackage{amssymb,amsmath,relsize}
\usepackage{enumerate}
\usepackage{url}
\usepackage{graphics}

\newcommand{\ovl}[1]{\overline{#1}}
\newcommand{\rar}{\rightarrow}
\newcommand{\subsEv}{\textup{ev}}
\newcommand{\subsOdd}{\textup{odd}}
\newcommand{\Cev}{C_{\subsEv}}
\newcommand{\Codd}{C_{\subsOdd}}
\newcommand{\pev}{P_{\subsEv}}
\newcommand{\podd}{P_{\subsOdd}}
\newcommand{\qev}{Q_{\subsEv}}
\newcommand{\qodd}{Q_{\subsOdd}}
\newcommand{\pevn}[1]{P_{\subsEv,#1}}
\newcommand{\poddn}[1]{P_{\subsOdd,#1}}
\newcommand{\mev}{M_{\subsEv}}
\newcommand{\modd}{M_{\subsOdd}}

\newcommand{\cat}{^\frown}
\newcommand{\norm}[2]{\left\|#2\right\|_{\mathsmaller{#1}}}
\newcommand{\nio}[1]{\norm{1\infty}{#1}}
\newcommand{\noi}[1]{\norm{\infty1}{#1}}
\newcommand{\ntwo}[1]{\norm{2}{#1}}
\newcommand{\none}[1]{\norm{1}{#1}}
\newcommand{\ninf}[1]{\norm{\infty}{#1}}
\newcommand{\ufirst}{u_{\mathsmaller{-}}}
\newcommand{\ulast}{u_{\mathsmaller{+}}}

\newtheorem{proposition}[theorem]{Proposition}
\newtheorem*{proposition*}{Proposition}

\usepackage{tikz}
\usetikzlibrary{trees,calc,positioning,decorations}

\newcommand{\tikzfrac}[2]{^{#1}\!/\!_{#2}}

\tikzstyle{level 1}=[level distance=1cm, sibling distance=1.2cm]
\tikzstyle{level 2}=[level distance=1cm, sibling distance=0.7cm]

\newcommand{\treeGenerate}{
  child { child child }
  child { child child }
}
   
\newcommand{\treeLabel}[8]{
   \node[anchor=south] at (#1-2-2) {#2};
   \node[anchor=south] at (#1-2-1) {#3};
   \node[anchor=south] at (#1-1-2) {#4};
   \node[anchor=south] at (#1-1-1) {#5};
   \node[anchor=east]  at (#1-2)   {#6};
   \node[anchor=west]  at (#1-1)   {#7};
   \node[anchor=north] at (#1)     {#8};
}

\begin{document}

\maketitle  

\begin{abstract}
  Symmetry of information states that $C(x) + C(y|x) = C(x,y) + O(\log C(x))$. 
  In~\cite{onlineComplexity} an online variant of Kolmogorov complexity is introduced and we show
  that a similar relation does not hold.
  Let the even (online Kolmogorov) complexity of an $n$-bitstring $x_1x_2\dots x_n$ be the 
  length of a shortest program that computes $x_2$ on input $x_1$, computes $x_4$ on input $x_1x_2x_3$, etc; 
  and similar for odd complexity.
  We show that for all $n$ there exists an $n$-bit $x$
  such that both odd and even complexity are almost as large as the Kolmogorov complexity of
  the whole string.
  Moreover, flipping odd and even bits to obtain a sequence $x_2x_1x_4x_3\dots$,
  decreases the sum of odd and even complexity to $C(x)$.
  Our result is related to the problem of inferrence of causality in timeseries.
\end{abstract}

\section{Introduction}
\label{sec:intro}

Imagine two people want to perform a two-person theater play. First suppose that the 
play consists of only two independent monologues each one performed by one player.
Before performing, the players must memorize their part of the play, and the total 
studying effort for the two players together can be assumed to be equal to the effort 
for one person to study the whole script. 

Now imagine a play consisting of a large 
dialogue where both players alternate lines. 
Each player only needs to study their half of the lines, and it is sufficient 
to remember each line only after hearing the last lines of the other player.
Thus each player needs only to remember their incremental amount of information 
in his lines, and this suggests the total studying effort might be close to the 
effort for one person to study the whole script.

However, it often happens that after studying only 
his own lines, an actor can reproduce the whole piece.
Sometimes actors just study the whole piece.
This suggests that studying each half of the lines can be as hard as studying everything. 
In other words, the total effort of both players together might be close to twice the
effort of studying the full manuscript.

\medskip
Can we interpret this example in terms of Shannon information theory? In the first
case, let a theater play be modeled by a probability density function $P(X,Y)$ where $X$ and $Y$ 
represent the two monologues. Symmetry of information states that $H(X) + H(Y|X) = H(X,Y)$, 
i.e. the information in the first part plus the new information in the second part equals the total
information. This equality is exact and can be extended to the interactive case where 
a similar additivity property remains valid, and this contrasts to the story above.

\medskip
An absolute measure of information in a string is given by its Kolmogorov complexity, 
which is the minimal length of a program on a universal Turing machine that prints the string.
See section~\ref{sec:defs} for formal definitions.
Symmetry of information for Kolmogorov complexity 
holds within logarithmic terms~\cite{ZvonkinLevin,BauwensAdditivity}: $C(x) + C(y|x) = C(x,y) + O(\log C(x,y))$.

For the interactive case, we need the online variant of Kolmogorov complexity introduced
in~\cite{onlineComplexity}. Let $\Cev(x)$ denote the length of a shortest program that 
computes $x_2$ on input $x_1$, computes $x_4$ on input $x_1x_2x_3$, etc.; and similar for $\Codd(x)$.
In the above example all $x_i$ with odd $i$ correspond to lines for the first player and the others
to the second.

In Theorem~\ref{th:onlineDouble}, we show that there exist infinitely many bitstrings $x$, such that both $\Cev(x)$ and
$\Codd(x)$ are almost as big as $C(x)$, in agreement with our example.  In Theorem~\ref{th:onlineMain},
we show that there exists $c>0$ such that $(\Cev + \Codd - C)(x) \ge c|x|$, i.e. the online asymmetry of
information can be large compared to the length of $x$.
Finally, we raise the question how large $(\Cev + \Codd - C)(x)$ can be in terms of $|x|$.
A more direct upper bound is $|x|/2 + O(1)$, and one can raise the question whether this is tight.
We show there exists a smaller one:
there exists $c > 0$ such that $(\Cev + \Codd - C)(x) \le (1/2 - c)|x|$ for all large $x$. 

\medskip

Our main result is stronger and is related to the problem of defining
causality in time series.
%Suppose an adaptive system is used to perform an
%iterative randomized task (for example consider the brain of a reading person). % who coordinates a hand to follow
%%a randomly moving projected point). 
%Furthermore assume that in the corresponding computational process, two different regions 
%$\mathcal X$ and $\mathcal Y$ are involved from which we obtain measurements 
Imagine there exists 
a complex system (e.g. a brain) and we make some measurements in two parts of it.
The measurements are represented by bitstrings $x$ (from some part $X$ of the brain)
and $y$ (from some part $Y$). We perform these measurements regularly and get a sequence of 
pairs 
\[
 (x_1, y_1), (x_2, y_2), \dots
\]
We assume that both parts are communicating with each other, however, the time resolution 
is not enough to decide whether $y_i$ is a reply to $x_i$ or vice versa.
However, we might compare the {\em dialogue complexity} $\Codd + \Cev$ of 
\[
 x_1, y_1, x_2, y_2, \dots
\]
and
\[
 y_1, x_1, y_2, x_2, \dots
\]
and (following Occam's Razor principle) choose an ordering that makes the dialogue complexity
minimal. We show that these complexities can differ substantially. 

\medskip
Questions of causality are often raised in neurology and economics. The notions of Granger causality
and information transfer reflect the idea of ``influence'' and our result implies a theoretical notion of
asymmetry of influence that does not need to assume a time delay to ``transport'' information 
between~$X$ and $Y$ in contrast to existing definitions~\cite{granger,geweke,Sch2000,overview2013}.\footnote{
  In the case of three or more
  timeseries there exist algorithms that infer directed information flows between some variables in
  some special cases where enough conditional independence exist among the variables, see~\cite[p. 19--20, 50]{Pearl}.
  In our example no independence is assumed.
  }

To understand why (current) practical algorithms need a time delay to make inferences about the direction of influence, 
consider two variables $X,Y$ with a joint probability density function $P(X,Y)$.
Using Shannon entropy, we can quantify the 
influence of $X$ upon $Y$ as $I(Y;X) = H(Y) - H(Y|X)$. Symmetry of information directly implies that 
this equals the influence of $Y$ upon $X$: $H(X) - H(X|Y) = H(X) + H(Y) - H(X,Y)$. 
In the online setting, mutual information is replaced by  
information transfer, which
is well studied in the engineering literature~\cite{braintransfer1,Sch2000,palus,rosenblum,Wint2005,overview2013,abdul2013mutual}.
For time delays $k$ and $l>k$ the information transfer from $\mathcal X$ to $\mathcal Y$ is given by 
\[
H(Y_{n} | Y_{n-l},\dots, Y_{n-1}) - H(Y_{n} | Y_{n-l},\dots, Y_{n-1}, X_{n-l},\dots X_{n-k}) \,,
 \]
(if this term is dependent on $n$, the sum is taken).
This quantification of causality coincides with Granger causality~\cite{granger,geweke} 
if all involved conditional distributions are Gaussian.

If we incorporate a time delay $k \ge 1$, the information transfers from $\mathcal X$ to $\mathcal Y$ 
and $\mathcal Y$ to $\mathcal X$ can be different.
%%This notion is used especially in economics and
%neurology to identify the main direction of flow of information in interactive processes. Moreover,
On the other hand, for $k=0$ they are always equal,
and this is a corollary of (the conditional version of) symmetry of information. 
In the offline case, a similar observation holds for algorithmic mutual information:
$C(x) - C(x|y) = C(y) - C(y|x) + O(\log C(x,y))$.\footnote{
  However, logarithmic deviations can appear, 
  if one considers prefix complexity, for example if $y$ is chosen to be a string consisting of $K(x)$ zeros. In this case, 
  it is known that for each $n$ there exist $n$-bit $x$ such that $K(K(x)) - K(K(x)|x) \le O(1)$ 
  while $K(x) - K(x|K(x)) \ge \log n - O(\log \log n)$. Moreover, this small error was exploited in
  an earlier and more involved proof of  Theorem~\ref{th:onlineMain}~\cite{BauwensPhd}.}
In the online setting, algorithmic mutual information can be generalized to algorithmic information
transfer. For an $n$-bit $x$ and $y$ the version without time delay is given by
\[
  IT(x \rar y) = C(y) - \Cev(x_1y_1\dots x_ny_n)\,.
  \]
We show that for all $\epsilon > 0$ there are infinitely
many pairs $(x,y)$ with $|x|=|y|$ and $C(x,y) \ge \Omega(|x|)$ 
such that $IT(x \rar y) \le \epsilon C(x,y)$ while $IT(y \rar x)$
exceeds $C(x,y) + O(1)$.
Hence, in contrast to Shannon information theory, 
significant online dependence of $x_i$ on $y_i$ might not imply
significant online dependence of $y_i$ on $x_i$.

%The second part of Theorem~\ref{th:onlineDouble}  implies 
%that the dialogue complexity $(\Codd + \Cev)(x)$ heavily depends on the 
%order in which the dialogue $x$ is presented. We remind the example from the introduction. 
%Imagine there exists 
%a complex system (e.g. a brain) and we make some measurements in two parts of it.
%The measurements are represented by bitstrings $x$ (from some part $X$ of the brain)
%and $y$ (from some part $Y$). We perform these measurements regularly and get a sequence of 
%pairs 
%\[
% (x_1, y_1), (x_2, y_2), \dots
%\]
%We assume that both parts are communicating with each other, however, the time resolution 
%is not enough to decide whether $y_i$ is a reply to $x_i$ or vice versa.
%However, we might compare the dialogue complexity of 
%\[
% x_1, y_1, x_2, y_2, \dots
%\]
%and
%\[
% y_1, x_1, y_2, x_2, \dots
%\]
%and (following Occam's Razor principle) choose an ordering that makes the dialogue complexity
%minimal.

\medskip
Warning: 
%Occam's Razor principle represents a general belief 
%and it might not be useful in some specific interpretations and contexts.
%But there are more specific objections: 
The example where influence (and causality) is asymmetric heavily uses 
that shortest models are not computable. 
%(They are partially computable. 
%Corresponding semimeasures are non-computable 
%but lower semicomputable, see further in Lemma~\ref{lem:computableDecomp}).
Decompression algorithms used in practice are always total (or can be extended to total ones).
On the other hand, if one wants to be practical, it is natural to not only consider total 
algorithms but algorithms that terminate within some reasonable time bound (say polynomial).
On that level non-symmetry may reappear, even for one pair of messages, 
which was not possible in our setting. 
For example suppose $x_1$ represents a pair of large primes and $y_1$ represents their product, then 
it is much easier to produce first $x_1$ and then $y_1$ then vice versa.

\medskip
Muchnik paradox is a result about online randomness~\cite{MuchnikOnline} 
that is related to our observations.
Consider the example from~\cite{onlineComplexity}:
in a tournament (say chess), a coin toss decides 
which player starts the next game. 
Consider the sequence $b_1, w_1, b_2, w_2,\dots $ of coin tosses and winners of subsequent games. 
This sequence might not be random (the winner might depend on who starts), 
but we would be surprised if the coin tossing depends on previous winners.
%Muchnik's paradox states that there exist sequences that are not Martin-L\"of random, 
%but for which both odd and even bits are online random. 

More precisely, a sequence is Martin-L\"of random if no lower semicomputable martingale 
succeeds on it. To define  randomness for even bits, 
we consider martingales that only bet on even bits,
i.e. a martingale $F$ satisfies $F(x0) = F(x1)$ if $|x0|$ is odd. 
The even bits of $\omega$ are {\em online random} if no lower semicomputable martingale succeeds that 
only bets on even bits. 
(In our example, coin tosses $b_i$ are unfair if a betting scheme
makes us win on $b_1w_1b_2w_2\dots$ 
while keeping the capital constant for ``bets'' on~$w_i$.)
In a similar way randomness for odd bits is defined.
Muchnik showed that there exists a non-random sequence for which both
odd and even bits are online random. 
Hence, contributed information by the odd and even bits does not ``add up''.
Muchnik's paradox does not hold for the online version of computable randomness (where martingales
are restricted to computable ones), 
and is an artefact of the non-computability of the considered martingales. 
%In a similar way, the asymmetry of online complexity results from the 
%non-computability of the constructed compression functions and would disappear if 
%we only considered computable compression functions 
%(or computable measures see Lemma~\ref{lem:computableDecomp} below). 

\medskip

The article is organised as follows: the next section presents definitions and results.
The subsequent three sections are devoted to the proofs:
%%%%%  *Arxiv* version vs stacs version
first theorems are reformulated using online semimeasures,
then lower bounds are proven, and finally (in Appendix  section~\ref{sec:upperbound}) the upper
bound is proven. 
In the next appendix we generalize  Theorem~\ref{th:onlineDouble} for online computation with more
machines and present a more involved proof of the main result 
with slightly better parameters. In the last appendix we generalize symmetry of information to a 
chain rule for online Kolmogorov complexity.
%%%%%  Arxiv version vs *stacs* version
%first theorems are reformulated using online semimeasures,
%and then lower bounds are proven. In the full version of the paper, which is available on 
%ArXiv, there are four appendices containing:
%a proof of the chain rule for online complexity, 
%the generalization of Theorem~\ref{th:onlineDouble} for online computation with more machines,
%a version of  Theorem~\ref{th:onlineMain} with a larger linear constant,
%and a full proof of the upper bound (Theorem~\ref{th:upperbound}). 

\section{Definitions and results}
\label{sec:defs}

Kolmogorov complexity of a string $x$ on an optimal machine $U$ is the minimal 
length of a program that computes $x$ and halts. 
More precisely, associate with a Turing machine a function $U$ that maps pairs of strings to
strings. The conditional Kolmogorov complexity is given by
\[
C_U(x|y) = \min \left\{ |p|: U(p,y) = x \right\} \,.
\]
This definition depends on $U$, but there exist a class of machines for which 
$C_U(x|y)$ is minimal within an additive constant for all $x$ and $y$. 
We fix such an optimal $U$, and drop this index, see~\cite{LiVitanyi,GacsNotes} for details. 
If $y$ is the empty string, we write $C(x)$ in stead of $C(x|y)$, and the 
complexity of a pair $C(x,y|z)$ is given by applying an injective computable pairing function 
to $x$ and $y$.

The {\em even (online Kolmogorov) complexity}~\cite{onlineComplexity} of a string $z$ is
%is the minimal length of a program that computes $x_2$
%on input $x_1$, computes $x_4$ on input $x_1x_2x_3$, \dots, computes $x_i$ on input $x_1\dots
%x_{i-1}$ for all $i \le n$; or more formal let $U$ be an optimal plain machine and
%let $|x|$ denote the length of $x$, then
\[
\Cev(z) = \min \left\{ |p|: U(p,z_1\dots z_{i-1}) = z_i \textup{ for all } i=2,4,\dots, \le|z| \right\}.
  \] 
Again, there exists a class of optimal machines $U$ for which $\Cev$ is minimal within an additive
constant and we assume that $U$ is such a machine.
% and the dialogue complexty in a 
%dialogue $x$ is given by the sum of odd and even complexity.
%Now again one raises the question whether communicating $x$ ``in pieces'' can always happen 
%by at most $C(x)$ bits, and if not, can an effective permutation of $x$ remove this
%excess?\footnote{Combining two programs in the definition of $\Codd$ and $\Cev$ one can compute 
%an extention of $x$, but not necesserily the length $|x|$ of $x$. On te other hand, a program in the
%definition of $C(x)$ always computes $|x|$.} For this reason it would be better to compare total
%information transfer with decision %complexity~\cite{
Note that $C(x|y) - O(1) \le \Cev(y_1x_1\dots y_nx_n) \le C(x) + O(1)$ for $n$-bit~$x$
and~$y$.
Let $\Cev(w|v)$ be the conditional variant. The chain rule 
for the concatenation $vw$ of strings $v$ and $w$ holds: $\Cev(vw) = \Cev(v) + \Cev(w|v) + O(\log (|v|))$, 
see Appendix~\ref{sec:chainrule}.
In a similar way~$\Codd(x)$ is defined. % and let
%$\Codd(z) + \Cev(z)$ be the {\em dialogue complexity}. It represents
%the minimal amount of information that needs to pass through a communication channel 
%such that both sources can reproduce the dialogue $z$.
A direct lower and upper bound for $\Codd + \Cev$ are\footnote{
  The $O(\log |x|)$ term could be decreased to $O(1)$ if we compared online complexity with
  decision complexity~\cite{ShenRelations} as in~\cite{onlineComplexity}. 
  However, plain and decision complexity differ by at most $O(\log |x|)$, and because we focus on
  linear bounds, we do not use this rare variant of complexity.}
\[
  C(z) - O(\log |z|) \le (\Codd+\Cev)(z) \le 2C(z) + O(1)\,.
  \]
  The lower bound is almost tight, for example if all even bits of $z$ are zero.
Surprisingly, the upper bound can also be almost tight and $\Codd + \Cev$ can change significantly after a
simple permutation of the bits.
%there exist bitwise dialogues such that both sources contribute an amount of information 
%close to the total information in the dialogue. However, if odd bits were communicated one
%timestep later (i.e. after the next even bit), then these odd bits contain no new information: 

\begin{theorem}\label{th:onlineDouble}
  For every $\varepsilon>0$ there exist $\delta>0$ and a sequence $\omega$ such that for large $n$
  \[
  \begin{array}{r} \Codd(\omega_1\dots\omega_n) \\
        \Cev(\omega_1\dots\omega_n) \end{array} 
  \ge  (1-\varepsilon)C(\omega_1\dots\omega_n) + \delta n \,.
  \]
  Moreover, for all even $n$
  \begin{eqnarray}
   \Codd(\omega_2\omega_1\dots\omega_{n}\omega_{n-1}) & = & C(\omega_1\dots\omega_{n}) 
   + O(\log n) \label{eq:double1} \\
   \Cev(\omega_2\omega_1\dots\omega_{n}\omega_{n-1}) & \le & O(1)\,.\label{eq:double2} 
 \end{eqnarray}
\end{theorem}

The first part implies 
\[
  \limsup_{|x|\rar\infty} \frac{\Codd(x) + \Cev(x)}{C(x)} \ge 2 \;,
\]
and by the upper bound $\Codd, \Cev \le C + O(1)$, this supremum equals $2$.
Recall the definition $IT(x \rar y) = C(y) - \Cev(x_1y_1\dots x_ny_n)$ for $x,y,n$ such that
$n=|x|=|y|$. 
Let $x = \omega_1\omega_3\dots\omega_{2n-1}$ and $y = \omega_2\omega_4\dots\omega_{2n}$,
Theorem~\ref{th:onlineDouble} implies
\begin{eqnarray*}
  IT(x \rar y) & \le & \varepsilon C(x,y) + O(1) \\
  IT(y \rar x) & = & C(x,y) + O(1) \,,
\end{eqnarray*}
(where $C(x,y) \ge \delta n - O(1)$).\footnote{
  For the first we use $C(y)
  \le C(\omega_{1\dots 2n}) = C(x,y)$ up to $O(1)$ terms. For the second $C(x,y) \ge C(x) \ge
  \Cev(y_1x_1\dots y_nx_n) = C(x,y)$,
  thus $C(x) = C(x,y)$, while $\Cev(y_1x_1\dots y_nx_n) \le O(1)$. Also,
  note that $C(\omega_{1\dots 2n})$ must exceed $\delta n$ because it exceeds
  $\Codd(\omega_{1\dots 2n}) \ge \delta n $, all up to $O(1)$ terms.}
  
Theorem~\ref{th:onlineDouble} can be generalized to dialogues between $k \ge 2$ machines, i.e. 
if $k$ sources need to perform a dialogue, it can happen that each source must contain
almost full information about the dialogue. Moreover, if the order is changed, the
``contribution'' of all except one source becomes computable.
Let the complexity of bits $i \bmod k$ be given by
  \[
  C_{i \bmod k}(x) = \min \left\{ |p|: U(p,x_1\dots x_{j-1}) = x_j \text{ for all } j = i, i+k, \dots, \le|x| \right\}.
  \]
For every $k$ and $\varepsilon>0$ there exist a $\delta > 0$ and a sequence $\omega$ such that for all $i \le k$ 
and large $n$
\[
  C_{i \bmod k}(\omega_1\dots \omega_{n})  \ge  (1-\varepsilon)C(\omega_1\dots\omega_{n}) + \delta n
  \]
  Moreover, for 
$\tilde{\omega} = \omega_k\omega_1\dots\omega_{k-1} \,\omega_{2k}\omega_{k+1}\dots\omega_{2k-1}\, \dots$
  for all $n$, and $i = 2\dots k$: 
\begin{eqnarray*}
  C_{1 \bmod k}(\tilde{\omega}_1\dots \tilde{\omega}_{n}) & =  & C(\omega_1\dots\omega_{n}) + O(\log n) \\
  C_{i \bmod k}(\tilde{\omega}_1\dots \tilde{\omega}_{n}) & \le & O(1)\,.
\end{eqnarray*}

\bigskip
In Theorem~\ref{th:onlineDouble} the difference between $C$ and $\Codd+\Cev$
is linear in the length of the prefix of $\omega$. One might wonder how big this difference can be.
A direct bound is $|x|/2 + O(1)$. Indeed, the odd complexity 
of $x$ is at most $C(x)$ hence
\[
\left( \Codd + \Cev\right)(x) - C(x) = ( \Codd(x) - C(x)) +
\Cev(x) \le O(1) + |x|/2 + O(1) \,.
\]
The next theorem shows that the difference can indeed be $c|x|$ for a significant~$c$.

\begin{theorem}\label{th:onlineMain}
  There exist a sequence $\omega$ such that for all $n$
  \[
  (\Codd + \Cev)(\omega_1\dots \omega_n)
  \ge n(\log \tfrac{4}{3})/2 + C(\omega_1\dots \omega_n) - O(\log n)\,.
  \]
  Moreover, Equations~\eqref{eq:double1} and~\eqref{eq:double2} are satisfied. 
\end{theorem}
In the appendix we show how the factor $(\log \tfrac{4}{3})/2$ can further be improved to $(\log
\tfrac{3}{2})/2 \approx 0.292$ at the cost of weakening~\eqref{eq:double1} and~\eqref{eq:double2}.
On the other hand, the upper bound $1/2$ can not be reached:
%We do not know whether this difference $n (\log \frac{3}{2})/2$ is maximal. 
%However, we show that it is at most
%$n(\tfrac{1}{2}-\varepsilon)$ for some $\varepsilon > 0$. 

\begin{theorem}\label{th:upperbound}
  There exist $\beta < \tfrac{1}{2}$ such that for large~$x$
  \[
  \left( \Cev + \Codd - C\right)(x) \le \beta |x|\,.
  \]
\end{theorem}
%\noindent In fact we can choose $\beta = 0.491$\,. 
\noindent
In summary, 
$
\tfrac{1}{2}\log \tfrac{3}{2} \le \lim \sup \frac{\left( \Cev + \Codd - C\right)(x)}{|x|} < \tfrac{1}{2}\,,
$
but the precise value of the $\limsup$ is unknown.

\section{Online semimeasures}

We show that the problem of constructing strings where additivity of online complexity 
is violated is equivalent to constructing lower semicomputable semimeasures that 
can not be factorized into ``odd'' and ``even'' online lower semicomputable semimeasures.
Before defining such semimeasures and reformulating Theorems~\ref{th:onlineDouble}--\ref{th:upperbound}, 
we recall the algorithmic coding theorem.

A (continuous) semimeasure $P$ is a function from strings to $[0,1]$ 
such that $P(x0) + P(x1) \le P(x)$ for all $x$. 
A real function $f$ on strings is lower semicomputable if the set of all pairs $(x,r)$ of strings
and rational numbers such that $f(x) \le r$ is enumerable.
There exist a maximal lower semicomputable semimeasure $M(x)$, i.e. a lower semicomputable that exceeds any other 
such semimeasures within a constant factor: $M(x) = \sum_i 2^{-i} P_i(x)$ 
for an enumeration $P_1$, $P_2$, \dots\, of all such semimeasures
(see \cite{GacsNotes,LiVitanyi,bookShenVereshchagin} for details).
The coding theorem~\cite[Theorem 4.3.4]{LiVitanyi} implies 
\[
 \log 1/M(x) = C(x) + O(\log C(x))\,.
\]

\noindent
An {\em even (online) semimeasure}~\cite{onlineComplexity} is a function from strings to $[0,1]$
such that for all~$x$
\begin{enumerate}[$i.$]
  \item $P(x0) + P(x1) \le P(x)$ if $|x0|$ is even, 
  \item $P(x0) = P(x1) = P(x)$ otherwise.%\footnote{
%    The relation with conditional semimeasures is 
%    as follows: every even semimeasure $P$ defines a conditional semimeasure $Q$ 
%    for $|x| \le |y|$ defined as 
%    $Q(x|y) = P(x_1y_1\dots x_{|x|}y_{|x|})$ and vice versa.
%}
\end{enumerate}
The coding theorem generalizes to the online setting.
\begin{theorem}[\cite{onlineComplexity}]\label{th:onlineCoding}
 There exist maximal even (respectively odd) semimeasures. All such semimeasures $\mev$ (resp. $\modd$) satisfy
 \[
  \log 1/\mev(x) = \Cev(x) + O\left(\log \Cev(x)\right).
 \]
\end{theorem}
\noindent
Let $\omega_{k\dots l} = \omega_k \dots \omega_l$.
Theorems~\ref{th:onlineDouble}, \ref{th:onlineMain} and \ref{th:upperbound} follow from

\begin{proposition}\label{prop:onlineDouble}
  For all $\varepsilon > 0$ and lower semicomputable odd and even online semimeasures $\qodd$ and $\qev$,  
  there exist $\delta$, a sequence $\omega$, a lower semicomputable semimeasure~$P$, 
  and a partial computable~$F$ such that for all~$n$
  \[
  (\qodd\qev)(\omega_{1\dots n}) 
  \le (1-\delta)^n P(\omega_{1\dots n})^{2-2\varepsilon}
  \]
  and $F(\omega_{1\dots2n},\omega_{2n+2}) = \omega_{2n + 1}$.
\end{proposition}

\begin{proposition}\label{prop:3_4th}
  For all lower semicomputable odd and even online semimeasures $\qodd$ and $\qev$,  
  there exist a sequence $\omega$, a lower semicomputable semimeasure~$P$, 
  and a partial computable~$F$ such that for all~$n$
  \[
    (\qodd\qev)(\omega_{1\dots 2n}) \le  (3/4)^{n} P(\omega_{1\dots 2n})
  \]
  and $F(\omega_{1\dots2n},\omega_{2n+2}) = \omega_{2n + 1}$.
\end{proposition}

\begin{proposition}\label{prop:upperbound}
  For all lower semicomputable semimeasures $Q$,
  there exist $\alpha > \sqrt{1/2}$
 and a family of odd and even semimeasures~$\poddn{n}$ and~$\pevn{n}$ uniformly lower-semicomputable in $n$, 
 such that for all~$x$ 
 \begin{equation}%\label{eq:goalup}
   \poddn{|x|}(x)\pevn{|x|}(x) \ge \alpha^{|x|}Q(x)/4\,.
 \end{equation}
\end{proposition}

\begin{proof}[Proof that Proposition~\ref{prop:upperbound} implies  Theorem~\ref{th:upperbound}.]
Choose $Q = M$ in  Proposition~\ref{prop:upperbound} and let for a sufficiently small $c>0$
\[
\podd(x) = c \left( \frac{1}{1^2} \poddn{1}(x) + \frac{1}{2^2} \poddn{2}(x) + \dots \right)\,.
\]
Note that $\podd$ is a lower semicomputable odd semimeasure and
by universality $\podd(x) \le O(\modd(x))$. Hence $-\log \modd(x) \le -\log
\poddn{|x|}(x) + O(\log |x|)$.  Similar for $\pev(x)$. 
By the online coding theorem we obtain up to terms $O(\log |x|)$,
\[
(\Codd + \Cev)(x) \le -\log \left( \poddn{|x|}(x)\pevn{|x|}(x) \right)  \le -|x|\log \alpha -\log Q(x) \,.
\]
Here, $-\log \alpha < 1/2$ and the last term 
is bounded by $-\log M(x) \le C(x) + O(\log |x|)$. The $O(\log |x|)$  can be removed for large $|x|$ 
by choosing $-\log \alpha < \beta < 1/2$. 
\end{proof}

\begin{proof}[Proof that Proposition~\ref{prop:3_4th} implies  Theorem~\ref{th:onlineMain}.]
  Choosing $\qodd = \modd$ and $\qev = \mev$, the first part is immediate 
  by the coding theorem and \eqref{eq:double2} follows directly from 
  the definition of even complexity. For any $x$ we have
  \[
   \Codd(x) - O(1) \le C(x) \le \Codd(x) + \Cev(x) + O(\log |x|)
  \]
  We obtain~\eqref{eq:double1} by applying $\Cev(x) \le O(1)$.
  \end{proof}

  \begin{proof}[Proof that Proposition~\ref{prop:onlineDouble} implies  Theorem~\ref{th:onlineDouble}.]
  For  Theorem~\ref{th:onlineDouble} we also apply  Proposition~\ref{prop:onlineDouble} with $\qodd =
  \modd$ and $\qev = \mev$ to obtain for some $\delta'>0$
  \[
  (\Codd + \Cev)(\omega_{1\dots 2n}) \ge (2 - 2\varepsilon)C(\omega_{1\dots 2n}) + \delta' n\,.
  \]
  Notice that $\Codd \le C + O(1)$, hence $\Cev(\omega_{1\dots 2n}) \ge (1-2\varepsilon) C(\omega_{1\dots
  2n}) + \delta' n$; and similar for $\Codd$. 
  Conditions  \eqref{eq:double1} and  \eqref{eq:double2} follow in a similar way as above.
\end{proof}
The generalization of Theorem~\ref{th:onlineDouble} mentioned in section~\ref{sec:defs} is shown in the appendix.
We remark that $P$ in these theorems can not be computable, 
this follows from the subsequent lemma.

\begin{lemma}\label{lem:computableDecomp}
  For every computable semimeasure $P$, 
  there exist computable odd and even online semimeasures 
  $\podd$ and $\pev$ such that $\podd\pev = P$.
\end{lemma}

\begin{proof}
  Let $\varepsilon$ be the empty string and let 
  $\podd(\varepsilon)=P(\varepsilon)$ and $\pev(\varepsilon)=1$.
  %(In fact, we can choose any factorization of $P(\varepsilon)$ in factors bounded by~$1$.)
  %The decomposition of $P(x)$ is defined iteratively: 
  Suppose that at some node $x$ we have defined $\podd(x)$ and $\pev(x)$ such 
  that $\podd(x)\pev(x) = P(x)$. Then $\podd$ and $\pev$ are defined
  on $2$-bit extensions of $x$ according to Figure~\ref{fig:computableDecomp} 
  for $\gamma = P(x)$ and $\alpha = \pev(x)$ 
  [our assumption implies $\podd(x) = \gamma/\alpha$]. Note that 
  $\podd$ and $\pev$ are indeed computable odd and even semimeasures
  and that $\podd\pev = P$.
\end{proof}

\begin{figure}
  \centering
  \begin{tikzpicture}[grow=up]
    \coordinate (root) \treeGenerate;
      \treeLabel{root}{$a$}{$b$}{$c$}{$d$}{$e$}{$f$}{$\gamma$}

    \node (eq) at (1.8,1) {\LARGE{$=$}};

    \coordinate[right=4.2 of root] (odd) \treeGenerate;
      \treeLabel{odd}{$e$}{$e$}{$f$}{$f$}{$e$}{$f$}{$\gamma$}
    \node  at ($(odd) +(-1.4,1)$) {\Large{$\tikzfrac{1}{\alpha} \cdot$}};
    \node (dot) at ($(odd) +(1.6,1)$) {\LARGE{$\cdot$}};

    \coordinate[right=3.5 of odd] (even) \treeGenerate;
    \treeLabel{even}{$\frac{a}{e}$}{$\frac{b}{e}$}{$\frac{c}{f}$}{$\frac{d}{f}$}{$1$}{$1$}{$1$}

    \node at ($(even) +(-1.4,1)$) {\Large{$\alpha\cdot$}};
  \end{tikzpicture}
  \caption{Decomposing semimeasures into odd and even ones.
    \label{fig:computableDecomp}
  }
\end{figure}
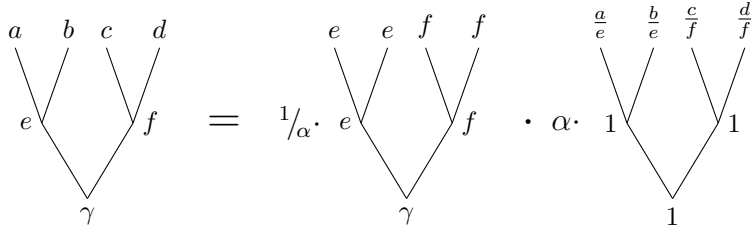

\section{Proofs of lower bounds}\label{sec:lowerbounds}

We start with Proposition~\ref{prop:3_4th}, and repeat it for convenience.

\begin{proposition*}
  For all lower semicomputable odd and even online semimeasures $\qodd$ and $\qev$,  
  there exist a sequence $\omega$, a lower semicomputable semimeasure~$P$, 
  and a partial computable~$F$ such that for all~$n$
  \[
    (\qodd\qev)(\omega_{1\dots 2n}) \le  (3/4)^{n} P(\omega_{1\dots 2n})
  \]
  and $F(\omega_{1\dots2n},\omega_{2n+2}) = \omega_{2n + 1}$.
\end{proposition*}

To develop some intuition, we first consider a game. 
The game is played between two players (Alice and Bob) who alternate turns. Alice maintains values
for $P(x)$ on $2$-bit $x$.  At each round she might pass or increase some values
as long as $\sum \{P(x): |x| = 2\} = 3/4$.  Bob maintains lower semicomputable
odd and even semimeasures $\qodd(x)$ and $\qev(x)$,
see figure~\ref{fig:game}.
Also Bob might pass or increase some values as long as the conditions of the definition of online
semimeasure are satisfied, (hence $\max\{ p+q, r+s, u+v \} \le 1$ in figure~\ref{fig:game}). 
Alice wins if in the limit $P(x) \ge \qodd(x)\qev(x)$ holds for some $x$ (i.e. 
if $P(00) \ge pr$ or $P(01) \ge ps$ or $P(10) \ge qu$ or $P(11) \ge qv$).

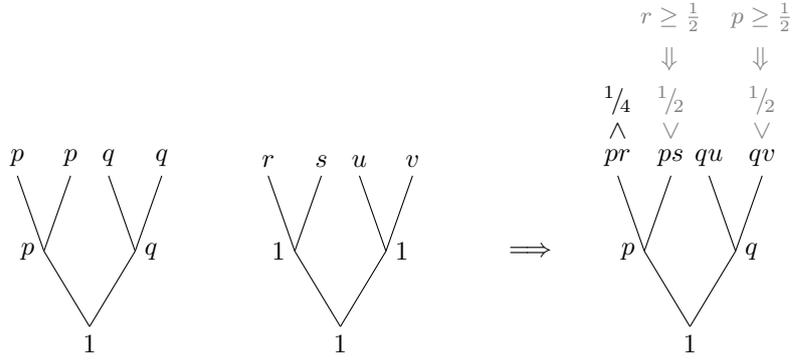
\begin{figure}
  \centering
  \begin{tikzpicture}[grow=up]

    \coordinate (odd) \treeGenerate;
    \treeLabel{odd}{$p$}{$p$}{$q$}{$q$}{$p$}{$q$}{$1$}

    \coordinate[right=3.3 of odd] (even) \treeGenerate;
    \treeLabel{even}{$r$}{$s$}{$u$}{$v$}{$1$}{$1$}{$1$}

    \node at (5.8,1) {$\Longrightarrow$};

    \coordinate[right=4.6 of even] (root) \treeGenerate;
    \treeLabel{root}{$pr$}{$ps$}{$qu$}{$qv$}{$p$}{$q$}{$1$}

      \path (root-2-2) -- node[midway,rotate=90,shift={(+0.1,0)}] {$>$} 
	+(0,1)  node {$\tikzfrac{1}{4}$};
      \path (root-2-1) -- node[midway,rotate=90,shift={(+0.1,0)},gray] {$<$} 
	++(0,1)  node[gray] {$\tikzfrac{1}{2}$} -- node[midway,rotate=90,gray] {$\Leftarrow$} 
	++(0,1.1)  node[gray] {\small{$r \ge \frac{1}{2}$}};
      \path (root-1-1) -- node[midway,rotate=90,shift={(+0.1,0)},gray] {$<$} 
	++(0,1)  node[gray] {$\tikzfrac{1}{2}$} -- node[midway,rotate=90,gray] {$\Leftarrow$} 
	++(0,1.1)  node[gray] {\small{$p \ge \frac{1}{2}$}};

%      \path (root-2-2) -- node[midway,rotate=90,shift={(+0.05,0)}] {$>$} 
%	+(0,1.2)  node {$\tikzfrac{1}{9}$};
%      \path  (root-1-1) -- ++(0,0.7) node  {$q's'$}
%      -- node[midway,rotate=90,shift={(-0.05,0)}] {$>$} 
%	++(0,0.8) node {$\tikzfrac{4}{9}$};
  \end{tikzpicture}
  \caption{Game for Proposition \ref{prop:3_4th} with $n=1$
    \label{fig:game}
  }
\end{figure}

In this game Alice has a winning strategy. She starts by putting $1/4$ at one leaf and 
zero at the others, say $P(00) = 1/4$.
Then she waits until Bob increases either $\qodd$ or $\qev$ above $1/2$ at this leaf
(thus $\qodd(0) = \qodd(00) > 1/2$ or $\qev(00)  >1/2$).
If none of this happens, Alice wins.
Otherwise if $\qodd(0) > 1/2$, she plays $P(11) = 1/2$ and if 
$\qev(00) >1/2$, she plays $P(01) = 1/2$. 
In the first case Alice wins because $\qodd(1) \le 1-\qodd(0) < 1/2$  and 
hence $\qodd(1)\qev(11) < 1/2$ and in the second case she wins because $\qev(01) \le 1-\qev(00) < 1/2$ and hence 
$\qodd(0)\qev(01) < 1/2$. Note that in both cases $\sum \{P(x): |x|=2\} = 1/2 + 1/4$, (and 
otherwise it is $1/4$) and Alice's condition is always satisfied. 
(Also note that the second bit of $x$ on which Alice wins is $1$ if $\qodd(0) > 1/2$ or
$\qev(00) > 1/2$. So for lower-semicomputable $\qodd$ and $\qev$, we can use this bit
to determine which inequality was first realized, and hence to compute the first bit of $x$. A similar
observation will be used to construct $F$ in the proof below.)

To show the proposition, we need to concatenate strategies for the game above to strategies for larger games. 
For this, it seems that the winning rule needs to be strengthened, and this makes either the winning
rule or the winning strategy for the small game complicated. Therefore, in the more concise proof
below, we gave a formulation without use of game technique. 

\begin{proof}
We construct $\omega_{1\dots2n}$ together with thresholds $o_{n}, e_{n}$ inductively.
Let $o_0 = e_0 = 1$. For $x$ of length $2n$,
consider the conditions $\qodd(x0) > o_{n}/2$ and $\qev(x00) > e_{n}/2$.
We fix some algorithm that enumerates $\qodd$ and $\qev$ from below and 
after each update tests both conditions. 
Let $O_x$ be the condition that $\qodd(x0) > o_{n}/2$ 
is true at some update and $\qev(x00) > e_{n}/2$ did not appear at any update strictly before;
and let $E_x$ be the condition that
$\qev(x00) > e_{n}/2$ is true after some update but $\qodd(x0) > o_{n}/2$ is false at the current
update (and hence at any update before).
Note that $O_x$ and $E_x$ cannot happen both. Let
\[
\left( \omega_{2n+1}\omega_{2n+2}, o_{n+1}, e_{n+1} \right) 
= \begin{array}{rllll}
  ( 11, & o_{n}/2, & e_{n} ) & \text{if $O_{\omega_{1\dots2n}}$ happens,} \\
  ( 01, & o_{n}, & e_{n}/2 ) & \text{if $E_{\omega_{1\dots2n}}$ happens,} \\
  ( 00, & o_{n}/2, &  e_{n}/2 ) & \text{otherwise.}
\end{array}
\]
By induction  it follows that 
$o_{n} \ge \qodd(\omega_{1\dots2n})$ and $e_{n} \ge \qev(\omega_{1\dots2n})$.
Indeed, this follows directly for $n=0$. For $n \ge 1$, consider the case where $O_{\omega_{1\dots 2n}}$
happens. Thus $\omega_{1\dots 2n+2} = \omega_{1\dots 2n+1}1$ and
\[
   \qodd(\omega_{1\dots 2n}1) \le \qodd(\omega_{1\dots 2n}) -
  \qodd(\omega_{1\dots 2n}0) \le o_n - o_n/2 = o_n/2 \,.
  \]
On the other hand, $\qev(\omega_{1\dots 2n+2}) \le \qev(\omega_{1\dots 2n}) \le e_n = e_{n+1}$.
The case where $E_{\omega_{1\dots 2n}}$ happens is similar, and the last one is direct.

It remains to define $F$ and $P$ such that $F(\omega_{1\dots{2n}}, \omega_{2n+2}) =
\omega_{2n+1}$ and
\[
  P(\omega_{1\dots2n}) = (4/3)^n o_{n} e_{n}\;.
  \]
Note that $\omega_{2n+2} = 1$ iff $O_{\omega_{1\dots2n}}$ or $E_{\omega_{1\dots2n}}$ happens,
and knowing that one of the events happens, we can decide which one and therefore also
$\omega_{2n+1}$.
Hence, given $\omega_{1\dots2n}$ and $\omega_{2n+2}$ we can compute $\omega_{2n+1}$ 
and this procedure defines the partial computable function $F$.

To define~$P$, observe that $\omega$ can be approximated from below: start with $\omega = 00\dots$, 
each time $O_{\omega_{1\dots2n}}$ (respectively $E_{\omega_{1\dots2n}}$) happens, 
change $\omega_{2n}\omega_{2n+1}$ from $00$ to $01$ (respectively to $11$), 
let all subsequent bits be zero, and repeat the process. 
Hence, for all $n$ and $2n$-bit $x$ at most 
one pair $(o_n,e_n)$ is defined which we denote as $(o_x,e_x)$.
Let $P(x)$ be zero unless $(o_x,e_x)$ is defined
in which case \[P(x) = (4/3)^{|x|/2} o_xe_x\,.\] Note that $P$ is lower semicomputable and the
equation above is satisfied.  Also, $P$ is a semimeasure: 
$P(\varepsilon) = (4/3)^0 \cdot 1 \cdot 1 = 1$, and
in all three cases we have
$\sum \{ o_{xbb'}e_{xbb'} : b,b' \in \{0,1\} \}  \le 3o_xe_x/4$ hence, 
$\sum \{ P(xbb') : b,b' \in \{0,1\} \} \le P(x)$. 
\end{proof}

\bigskip
The proof of Proposition \ref{prop:onlineDouble} follows the same structure. 
\begin{proposition*}
  For all $\varepsilon > 0$ and lower semicomputable odd and even online semimeasures $\qodd$ and $\qev$,  
  there exist $\delta$, a sequence $\omega$, a lower semicomputable semimeasure~$P$, 
  and a partial computable~$F$ such that for all~$n$
  \[
  (\qodd\qev)(\omega_{1\dots n}) 
  \le (1-\delta)^n P(\omega_{1\dots n})^{2-2\varepsilon}
  \]
  and $F(\omega_{1\dots2n},\omega_{2n+2}) = \omega_{2n + 1}$.
\end{proposition*}

\begin{proof}
We first consider the following variant for the game above on strings of length two. 
Alice should satisfy the weaker condition $\sum \{P(x): |x|=2\} \le 1-\delta$, 
where $\delta \ll \varepsilon$ will be determined later.
She wins if 
\[
 (\podd\pev)(x) \le \left(P(x)\right)^{2-2\varepsilon} 
\] 
for some $x$.
The idea of the winning strategy is to
start with a very small value somewhere, say $P(00) = \delta$. 
If $\varepsilon = 0$ then Bob could reply with $\qodd(0) = \qev(00) = \delta$, 
(in fact he could win by always choosing $\qodd(x)=\qev(x)=P(x)$). 
For $\varepsilon > 0$ and $\delta \ll \varepsilon$ one of the online semimeasures should 
exceed $\delta^{1-\varepsilon} = k\delta$ for $k=\delta^{-\varepsilon}$. 
$k$ can be arbitrarily large if $\delta \ll \varepsilon$ is chosen sufficiently small. 
At his next move, (as before), Alice puts all his remaining measure, i.e. $1-2\delta$ in a leaf 
that does not belong to a branch where the corresponding online semimeasure is large. 
Note that $1-2\delta$ is close to $1$  
and taking a power $2 \ge 2-2\varepsilon$ we see that Bob 
needs at least $1-4\delta$ in each online semimeasure, but he already used $k\delta$ in one of them.

More precisely, the winning strategy for Alice is to set $P(00) = \delta$ and wait until 
$\qodd(0) > \delta^{1-\varepsilon}$ or $\qev(00) >\delta^{1-\varepsilon}$. 
If these conditions are never satisfied, then Alice wins on $x=00$.
Suppose at some moment Alice observes that the first condition holds, 
then she plays $P(11) = 1-2\delta$, in the other case she plays $P(01)  = 1-2\delta$. 
Afterwards she does not play anymore. Note that $\sum \{P(x) : |x|=2\} \le 1-\delta$.
We show that Alice wins. Assume that $\qodd(0) > \delta^{1-\varepsilon}$ (the other case is similar). 
We know that $\qev(11) \le 1$ hence if Alice does not win, 
this implies $\qodd(1) > (1-2\delta)^{2-2\varepsilon}$.
This is lower bounded by $(1-2\delta)^2 \ge 1 - 4\delta$.
We choose $\delta = 2^{-2/\varepsilon}$. This implies
\[ 
  \delta^{1-\varepsilon} = 2^{-(2/\varepsilon)(1-\varepsilon)} 
  = 2^{-2/\varepsilon + 2} = 4\delta.
\]
Hence $\qodd(0) + \qodd(1) > 4\delta + (1-4\delta) = 1$ and 
Bob would violate his restrictions. Therefore Alice wins.
For later use notice that in the first case our argument implies $\qodd(1) \le (1-2\delta)^{2-2\varepsilon}$.

In a similar way as before we adapt Alice's strategy  to an inductive 
construction of $\omega$ and $P$: 
let $O_x$ and $E_x$ be defined as before using conditions 
$\qodd(x0) > o_{n} \delta^{1-\varepsilon}$ and 
$\qev(x00) > e_{n} \delta^{1-\varepsilon}$.
Let $\beta = (1-2\delta)^{2-2\varepsilon}$ and let $\omega, o_{n}$ and $e_{n}$ be given by
\[
  \left( \omega_{2n+1}\omega_{2n+2}, o_{n+1}, e_{n+1} \right) 
  = \begin{array}{rllll}
    ( 11, & o_{n}\beta, & e_{n} ) & \text{if $O_{\omega_{1\dots2n}}$ happens,} \\
    ( 01, & o_{n}, & e_{n}\beta ) & \text{if $E_{\omega_{1\dots2n}}$ happens,} \\
    ( 00, & o_{n}\delta^{1-\varepsilon}, &  e_{n}\delta^{1-\varepsilon}) & \text{otherwise.}
  \end{array}
\]
This implies
$o_{n} \ge \qodd(\omega_{1\dots2n})$ and $e_{n} \ge \qev(\omega_{1\dots2n})$.
$F$ is defined and shown to satisfy the condition in exactly the same way.
It remains to construct $P$ such that 
\[
(1-\delta)^n P(\omega_{1\dots2n}) =  \left( o_{n} e_{n} \right)^{1/(2-2\varepsilon)}\;,
\]
(the proposition follows after rescaling $\delta$).
In a similar way as before $o_x$ and $e_x$ are defined and let 
\[ 
P(x) = (1-\delta)^{-|x|/2} (o_xe_x)^{1/(2-2\varepsilon)} \,.
\]
To show that $P$ is indeed a semimeasure observe that $\sum \{P(xbb'): b,b' \in \{0,1\} \}$ 
\begin{eqnarray*}
  && = (1-\delta)^{-|x|/2 - 1}\sum \{ \left(o_{xbb'}e_{xbb'}\right)^{1/(2-2\varepsilon)} : b,b' \in \{0,1\} \} \\ 
  && \le (1-\delta)^{-|x|/2 - 1} \left(\beta^{1/(2-2\varepsilon)}  + \delta \right) \left(o_xe_x\right)^{1/(2-2\varepsilon)} \,,
\end{eqnarray*}
and because $\beta^{1/(2-2\varepsilon)} = 1-2\delta$ this equals
\[
  = (1-\delta)^{-|x|/2}  \left(o_xe_x\right)^{1/(2-2\varepsilon)} = P(x)\,.
  \qedhere
  \]
\end{proof}

\subsubsection*{Acknowledgements}
  The author is grateful to Alexander Shen, Nikolay Vereshchagin, Andrei Romashchenko, Mikhail
  Dektyarev,  Ilya Mezhirov and Emmanuel Jeandel for extensive discussion and many useful suggestions.
  I also thank Ilya Mezhirov for implementing clever code to study some games.
  Especially thanks to Alexander Shen for encouragement after presenting earlier results and
  for arranging funding by grant NAFIT ANR-08-EMER-008-01. The author is also grateful to Mathieu 
  Hoyrup who arranged a grant under which the work was finalized. 
\bibliography{eigen,kolmogorov,practCausalities,statisticalCausalities}

\appendix

\section{Proof of the upper bound: Theorem~\ref{th:upperbound}}
\label{sec:upperbound}

It remains to prove  Proposition~\ref{prop:upperbound}, which we repeat.
\begin{proposition*}
  For all lower semicomputable semimeasures $Q$,
  there exist $\alpha > \sqrt{1/2}$
 and a family of odd and even semimeasures~$\poddn{n}$ and~$\pevn{n}$ uniformly lower-semicomputable in $n$, 
 such that for all~$x$ 
 \begin{equation}%\label{eq:goalup}
   \poddn{|x|}(x)\pevn{|x|}(x) \ge \alpha^{|x|}Q(x)/4\,.
 \end{equation}
\end{proposition*}

Consider the game defined before the proof of
Proposition~\ref{prop:3_4th} on $2$-bit $x$, where we replace $3/4$ by 
any real number $\alpha^2$.
Thus, Alice enumerates non-decreasing values $P(x)$ such that $P(x) \le \alpha^2$
and Bob enumerates $\podd(x)$ and $\pev(x)$ for all two-bit $x$. 
Bob wins if $\podd(x)\pev(x) > P(x)$ for all~$x$.
It can be shown that  Bob has a winning strategy in this game iff $\alpha^2 < 2/3$, and 
his strategy is uniformly computable in~$\alpha$. (Otherwise Alice wins, see
appendix~\ref{sec:2_3th}.)
Proposition~\ref{prop:upperbound} would follow for all $\alpha^2 < 2/3$ if 
induction schemes for these strategies exist. 
Unfortunately, we do not know such schemes.
Therefore, a game is formulated with a stronger winning condition for Bob that 
allows to concatenate winning strategies. We show a winning strategy for
some $\alpha^2$ (significantly below $2/3$ and) above~$1/2$.

\begin{proof}
 Let $\varepsilon>0$ be small to be determined later.
 For every $n$ and every lower semicomputable semimeasure $P_n$ defined on all $n$-bit $x$
 we construct lower semicomputable $\podd$ and $\pev$ such that
 $(\podd\pev)(x) \ge (1/\sqrt{2} + \varepsilon/2)^n P_n(x)/4$ (on all $n$-bit $x$).
 Moreover, our construction is uniform in $n$ and $P_n$ and this implies  
 Proposition~\ref{prop:upperbound}.

 We will  represent $\sqrt{P_n}$ by a $2^n$-dimensional real vector $u$ such that
 $u_1$, $u_2$, \dots, $u_{2^n}$ equal the values of $\sqrt{P_n(x)}$ on all $n$-bit $x$ in lexicographic order.
 Note that the definition of semimeasure implies $\ntwo{u} \le 1$.

 We construct online semimeasures $\podd$ and $\pev$ from $2^n$-dimensional
 vectors~$o$ and~$p$ by defining $\podd$ and $\pev$ to be the smallest odd and even
 online semimeasures whose values on all $n$-bit $x$ in lexicographic order do not exceed
 $o_1$, \dots, $o_{2^n}$ and $p_1, \dots, p_{2^n}$.
 Note that such $\podd$ and $\pev$ can be computed from $o$ and $p$ because for fixed $n$ only
 finitely many binary max and sum operations appear in the computation.

 To see whether some $2^n$-dimensional vectors~$o$ and~$p$ indeed define such semimeasures,
 consider all large enough functions $\pev$ and $\podd$ satisfying conditions $i$ and $ii$ 
 in the definition of online semimeasures. 
 Let us derive these minimal values for $\podd(\varepsilon)$ and $\pev(\varepsilon)$ (which should
 be at most one).
 For the case $n=1$, $o$ and $p$ are $2$-dimensional 
 and the minimal root values are $\ninf{o}$ and $\none{p}$.
 For $n>1$, note that in the tree representation of $\podd$ and $\pev$, 
 a value of a node is lower bounded by either the max or the sum 
 of its child values (see figure~\ref{fig:norms}).
 Hence, we define the norms $\noi{.}$ and $\nio{.}$ inductively:
 for a $1$-dimensional vector $u$ let $\noi{u} = \nio{u} = u_1$.
 For an even dimensional vector $u$ let let $\ufirst$ and $\ulast$ denote the first and last 
 halve of indices of $u$. 
 For a $2^{n+1}$-dimensional vector $u$ let
 \begin{eqnarray*}
   \noi{u} &=& \ninf{\nio{\ufirst}, \nio{\ulast}}\\
   \nio{u} &=& \none{\noi{\ufirst}, \noi{\ulast}} \;.
 \end{eqnarray*}
 Note that $\noi{\cdot}$ and $\nio{\cdot}$ indeed define norms.
 We construct functions $o$ and~$p$ such that 
 \begin{enumerate}
   \item \;\;$\noi{o(u)} \le \ntwo{u}$, \label{cond:o}
   \item \;\;$\nio{p(u)} \le \ntwo{u}$, \label{cond:p}
   \item \;\;$o(u)p(u) \ge (1/\sqrt{2} + \varepsilon/2)^n u^2$, \label{it:mult} \label{cond:winning}
 \end{enumerate}
 (vectors are multiplied point wise).
 We do this in an effective way and guarantee that if some coordinates of $u$ are non-decreasingly
 updated, then also $o$ and $p$ have non-decreasing updates. These updates might depend on 
 the history of updates for $u$. %so $u$ in the formulas above represents a finite series of vectors.
 By the above discussion, the functions $o$ and $p$ define the requested lower semicomputable online semimeasures. 
 
  \begin{figure}
    \centering
    \begin{tikzpicture}[grow=up]
      \tikzstyle{level 1}=[level distance=1cm, sibling distance=4cm]
      \tikzstyle{level 2}=[level distance=1cm, sibling distance=2cm]
      \tikzstyle{level 3}=[level distance=1cm, sibling distance=1cm]

      \node (norms) {+}
      child {
	 node {M}
	 child {node {+} child {node {$u_8$}} child {node {$u_7$}} }
	 child {node {+} child {node {$u_6$}} child {node {$u_5$}} }
	 }
      child {
	 node {M}
	 child {node {+} child {node {$u_4$}} child {node {$u_3$}} }
	 child {node {+} child {node {$u_2$}} child {node {$u_1$}} }
	 };
     \end{tikzpicture}
     \caption{$\nio{u}$ for $u \in \mathbb{R}^8$}
     \label{fig:norms}
  \end{figure}

 \medskip
 It remains to construct $o$ and $p$. For $1$-dimensional $u$ (i.e. $n=0$) we choose
 $o(u)=p(u)=u$. Clearly the conditions are satisfied. 
 We explain the induction step. Suppose the case $n = 1$ is solved. In fact, for small $\varepsilon$ 
 our solution will be approximately $o(u) = (1+a_u\varepsilon)u$ and $p(u) = (1+b_u\varepsilon)u/\sqrt{2}$ 
 for some (piecewise constant) rational numbers $a_u$ and~$b_u$ such that $a_u+b_u \ge 1$. 
 %For $u$ of even dimension, let $\ufirst$ and $\ulast$ denote the first and last halves of $u$.
 Let $^\frown$ denote concatenation; thus $u = \ufirst^\frown\ulast$. 
 For the (approximate) induction step, first evaluate
 %The idea of the induction step %to compute $o(u)$ and $p(u)$ is 
 $o(\ufirst), p(\ufirst), o(\ulast), p(\ulast)$ and 
 $v_u = [\ntwo{\ufirst}, \ntwo{\ulast}]$.
 Then combine it using the factors from the 
 $n=1$ case: 
 $o(u) = (1+a_{v_u}\varepsilon)(p(\ufirst)\cat p(\ulast))$ 
 and similar for $p(u)$. 
 %This is essentially the induction scheme that will be used.
 %let $u = y\cat z$ with $|y| = |z|$; with $\cat$ being the concatenation of vectors. 

\begin{figure}
  \centering
  \begin{tikzpicture}[grow=up,scale=0.7, every node/.style={transform shape}]
    \tikzstyle{level 1}=[level distance=1cm, sibling distance=8cm]
    \tikzstyle{level 2}=[level distance=1cm, sibling distance=4cm]
    \tikzstyle{level 3}=[level distance=1cm, sibling distance=2cm]

    \node (norms) 
    {$\overline{p}(v_{u}, \searrow\cat\swarrow)$} 
    child {
       node 
	{$\overline{o}(v_{u_{5\dots 8}}, \searrow\cat\swarrow)$} 
       child 
       child
       }
    child {
       node 
    {$\overline{o}(v_{u_{1\dots 4}}, \searrow\cat\swarrow)$} 
       child {
          node {$\overline{p}(u_3u_4, u_3u_4)$} 
	  child {node {$u_4$}} 
	  child {node {$u_3$}} 
	 }
       child {
          node {$\overline{p}(u_1u_2, u_1u_2)$} 
	  child {node {$u_2$}} 
	  child {node {$u_1$}} }
       };
   \end{tikzpicture}
   \caption{$p(u)$ for $u \in \mathbb{R}^8$}
   \label{fig:induction}
\end{figure}
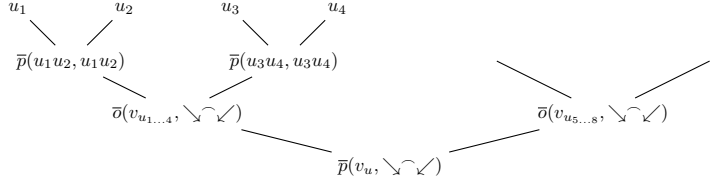

More precisely, we construct functions
$\ovl{o}(v,u)$ and $\ovl{p}(v,u)$ defined for $2$-dimensional $v$ and real vectors $u$ of any
dimension.  (The functions are obtain from the $n=1$
step.) %In fact for $2$-dimensional $u$ we have $o(u) = \ovl{o}(u,u)$ and $p(u) = \ovl{p}(u,u)$.)
$o$ and $p$ for $u$ of dimension $2^n \ge 2$ are inductively defined by
 \begin{eqnarray*}
   o(u) &=&  \ovl{o}\left(v_u, p(\ufirst)^\frown p(\ulast) \right) \\
   p(u) &=&  \ovl{p}\left(v_u, o(\ufirst)^\frown o(\ulast) \right) \,.
 \end {eqnarray*}
 The computation of $p(u)$ is illustrated in figure~\ref{fig:induction} for $n=3$.
% assume that in the $n=1$ step we have constructed $o(u)$ and $p(u)$ (for 2-dimensional $u$). 
% The functions can be defined more generally as 
 The functions $\ovl{o}$ and $\ovl{p}$ satisfy the following properties:
 \begin{itemize}
   \item 
     if $\noi{\ufirst}\le v_0$ and $\noi{\ulast} \le v_1$, then $\nio{\ovl{p}(v,u)} \le \ntwo{v}$,
   \item 
     if $\nio{\ufirst}\le v_0$ and $\nio{\ulast} \le v_1$, then $\noi{\ovl{o}(v,u)} \le \ntwo{v}$,
   \item $(\ovl{o}\,\ovl{p})(v,u) \ge (1/\sqrt{2} + \varepsilon/2)u^2$,
 \end{itemize}
 Remind that $v_u = [\ntwo{\ufirst}, \ntwo{\ulast}]$, thus $\ntwo{v_u} = \ntwo{u}$.
 By induction, $o$ and $p$ satisfy properties 1 and 2, and property 3 follows directly.

 \medskip
 We now construct $o(u)$ and $p(u)$ for $2$-dimensional $u$ such that conditions [1]-[3] are
 satisfied; afterwards we finish the proof by generalizing the construction to obtain
 $\ovl{o}$ and $\ovl{p}$ satisfying the conditions of the induction step. 

 Let us try the following solution: $o(u) = u$ and $p(u) = \alpha u$ for some $\alpha > 1/\sqrt{2}$. 
 Clearly, [1] and [3] are satisfied (for some $\varepsilon > 0$).
 Consider the second condition:
 $u_0 + u_1 \le \sqrt{u_0^2 + u_1^2}$ ($2$-dimensional vectors are labeled
 by $0$ and $1$, rather than $1$ and $2$).
 For $u_0 = u_1$ this can only hold if $\alpha \le
 1/\sqrt{2}$, violating the assumption. On the other hand, in this special case [3] is satisfied within 
 a large margin.
 See figure~\ref{fig:naive} for an illustration of the conditions. 
 Hence let us try the solution $o(u) = \alpha u$ and $p(u) = u/\sqrt{2}$ for some $\alpha > 1$.
 Now [2] and [3] are satisfied. Choose $u_0 = 0$ and $u_1 = 1$, 
 the first condition implies $\alpha \le 1$, again violating the assumption. 
 On the other hand, [2] is satisfied within a large margin.

\begin{figure}
  \centering
  \begin{tikzpicture}
    \draw[<-] (2.5,0) node[anchor=north] {$u_0$} -- (0,0);
    \draw[<-] (0,3) node[anchor=east] {$u_1$} -- (0,0);
    \fill  (0,1.9) circle (2pt);
    \draw  (0,2.4) circle (3pt);
    \node  at (0.03,1.2) {+};

    \fill  (1.7,1.7) circle (2pt);
    \draw (0,0) -- (1.7,1.7);
    \draw  (1.8,2) circle (3pt);
    \draw[gray,dashed] (1.8,1.85) -- (1.8,0);
    \draw[gray,dashed] (1.65,2) -- (0,2);
    \node  at (1.5,0.8) {+};
    \draw[gray] (1.5,0.8) -- (1.5,0.03) -- (0,0.03);
   \end{tikzpicture}
   \caption{$p$  and $o$  for $u \in \mathbb{R}^2$ ($p,o$ and $u$ 
   are denoted as $+,\circ$ and $\bullet$).}
   \label{fig:naive}
\end{figure}
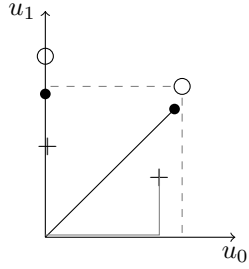

 These observations suggest to construct $o$ and $p$ differently according to the region where $u$
 is located: if $u_0 \approx u_1$ then 
 $p$ should be relatively small while $o$ is large and if $u_0/u_1$ or $u_1/u_0$ is small, then  $p$
 should be relatively large while $o$ is small.
 In the balanced region (i.e. where $u_0 \approx u_1$), 
 $p$ grows slowly at updates of $u$, so that at the critical line $u_0=u_1$ 
 the excess from the unbalanced region is compensated and $p$ is relatively small. 
 In order to satisfy [3], $o$ must grow faster.
 In the unbalanced region,  $o$ grows slower so that if $u_0/u_1$ or $u_1/u_0$ become small, 
 the excess from the balanced stage is compensated. To satisfy~[3] we need that $p$ grows fast. 
 Despite the simple idea, the precise calculations seem laborious.

 Let $u = [a,b]$. 
 $u$ is {\em balanced} iff $a\sqrt{20\varepsilon} \le b \le a/\sqrt{20\varepsilon}$,
 otherwise $u$ is  {\em unbalanced}.
 At each stage let the values of $o,p,u,\dots$ (we drop arguments $u$) 
 after the previous stage be denoted as
 $o_o, p_o, u_o, \dots$, where a subscript o abbreviates ``old''.
 For unbalanced $u$ let
 \begin{eqnarray}
   p &=&  (1+2\varepsilon)u/{\sqrt{2}} \nonumber \\
   o &=& \max \left\{ o_o, (1-\varepsilon)u \right\}\,, \label{eq:def_op1}
 \end {eqnarray}
 and if $u$ is balanced let
 \begin{eqnarray}
   p &=& \max \left\{ p_o, (1-4\varepsilon)u/{\sqrt{2}} \right\} \nonumber\\
   o &=& (1+5\varepsilon)u\,.\label{eq:def_op2}
 \end {eqnarray}
 In these definitions the max functions guarantee that $o$ and $p$ are non-decreasing in those
 regions where they grow slower.
 For small $\varepsilon$ condition [\ref{cond:winning}] is always satisfied. We need to check
 [\ref{cond:o}] and [\ref{cond:p}] in 
 the unbalanced an balanced stage. %Our argument will be recycled in the next part.

 1) Balanced $\ninf{o} \le \ntwo{u}$, i.e. $o_0 \le \ntwo{u}$ and $o_1 \le \ntwo{u}$. 
 We show the first one (the other is similar):
 \begin{equation}\label{eq:1}
  (1+5\varepsilon)a \le \sqrt{a^2 + b^2}\,, 
 \end{equation}
 i.e. $1+5\varepsilon \le \sqrt{1 + (b/a)^2}$. Because $b^2 \ge 20\varepsilon a^2$ the right hand side 
 is at least $1 + 10\varepsilon - O(\varepsilon^2)$, and this exceeds the left hand side for small
 $\varepsilon$.

 2) Unbalanced $\none{p} \le \ntwo{u}$:
 \begin{equation}\label{eq:2}
   (1+2\varepsilon)(a+b)/\sqrt{2} \le \sqrt{a^2 + b^2}\,.
 \end{equation}
 Rearranging: 
 \[
   (1+2\varepsilon)/\sqrt{2} \le  \frac{\sqrt{a^2 + b^2}}{a+b} \,. 
   \]
 Assume the case $b/a \le 20\varepsilon$, the other case is similar.
 Dropping $b^2$ in the right-hand, we have $1/(1+b/a)$ 
 which is bounded by $1-\sqrt{20\varepsilon} + O(\varepsilon)$.
 Hence it exceeds the left hand for small $\varepsilon$.

 3) Unbalanced $\ninf{o} \le \ntwo{u}$, i.e. $o_0 \le \ntwo{u}$ and 
 $o_1 \le \ntwo{u}$. Only the first is shown, the second is similar.
 Suppose no balanced stage has happened, then $o \le (1-\varepsilon)u$. Otherwise,
 let $u_o$ be the value of $u$ at the last balanced stage. In our upperbound for $o$ 
 there might be missing at most an excess 
 $(1+5\varepsilon)u_o - (1-\varepsilon)u_o = 6\varepsilon u_o$.
 Thus, the equation becomes 
 \begin{equation}\label{eq:3}
 (1-\varepsilon)a + 6\varepsilon a_o \le \sqrt{a^2 + b^2}.
 \end{equation}
 Note that $b \ge b_o \ge a_o\sqrt{20\varepsilon}$, hence for $x = a_o/a \le 1$ 
 this follows from 
 \begin{equation*}%\label{eq:3}
  1-\varepsilon + 6\varepsilon x \le \sqrt{1 + 20\varepsilon x^2} = 1 +
 10\varepsilon x^2 + O(\varepsilon^2)
 \end{equation*}
 i.e. $0 \le 10x^2 - 6x + 1 + O(\varepsilon)$. The discriminant is $36 - 40 + O(\varepsilon)$, 
 hence the inequality holds for small $\varepsilon$.
 
 4) Balanced $\none{p} \le \ntwo{u}$. In a similar way as before, we determine the excesses 
 at the last unbalanced stage and the condition becomes
 \begin{equation}\label{eq:4}
  (1-4\varepsilon)(a+b)/\sqrt{2} + 6\varepsilon(a_o + b_o)/\sqrt{2} \le \sqrt{a^2 + b^2}\,.
 \end{equation} 
 Suppose that $b_o \le a_o\sqrt{20\varepsilon}$ 
 (the other case is similar) and let $x = a/a_o$ and $y = b/a_o$, thus $x \ge 1$. 
 The equation is
 \[
  (1 - 4\varepsilon)(x+y) + 6\varepsilon + O(\varepsilon^{3/2}) \le \sqrt{2x^2 + 2y^2}\,.
 \]
 First suppose $y \ge 1$, thus $x + y \ge 2$. Observe that for varying $x$ and $y$ 
 the left hand side only depends on $x+y = z$. Note that $z \ge 2$. The smallest value in the right
 hand side is obtained 
 for $x=y=z/2$. 
 %Because we need to check the equation for all values $x \ge 1$ and $y \ge 1$, 
 %it indeed suffices to choose $x=y$ (in the case $y \ge 1$).
 In this case, terms without $\varepsilon$ cancel and the equation becomes 
 $-4z + 6 + O(\sqrt{\varepsilon}) \le 0$; which follows from $z \ge 2$.

 Suppose $y \le 1$. The slope of $x$ in the left-hand is $1-4\varepsilon$, and the slope in the
 right hand side is $2x/\sqrt{2x^2+2y^2} = \sqrt{2/(1+(y/x)^2} \ge 1$ (because $y \le 1 \le x$). 
 Hence, it suffices to 
 check the equation for $x=1$ (remind that $x \ge 1$): 
 $y^2(1+8\varepsilon) + y(2+10\varepsilon) + 1 + 8\varepsilon + O(\varepsilon^{3/2}) \ge 0$. 
 The discriminant is proportional to $(1 + 5\varepsilon)^2 - (1 + 8\varepsilon)^2 +
 O(\varepsilon^{3/2}) < 0$.
 Hence, $o$ and $p$ satisfy all conditions.
 
 A numerical search to find a maximal $\varepsilon$ satisfying equations \eqref{eq:1}--\eqref{eq:4} 
 shows that   Theorem~\ref{th:upperbound} holds for $\beta = 0.491$.

 \smallskip
 Note that in the construction of $o$ and $p$, we used a $2$-dimensional $u$ both 
 to determine whether a stage is balanced or unbalanced and to be a linear factor in 
 equations \eqref{eq:def_op1} and \eqref{eq:def_op2}.
 $\ovl{o}(v,u)$ and $\ovl{p}(v,u)$ are obtained by using a seperate $2$-dimensional vector 
 $v = [a,b]$ for the first purpose and $u$ for the second purpose. Now $u$ can have any dimension.
 It remains to show the conditions for induction. 
 Remark that all conditions  \eqref{eq:1}-\eqref{eq:4} can be written as 
 $c_1a + c_2b + c_3a_o + c_4b_o \le \ntwo{[a,b]}$ where $c_i \ge 0$ for $i = 1,\dots,4$. 
 Consider condition [\ref{cond:p}] ([\ref{cond:o}] is analogue): for coefficients $c_1\dots c_4$
 corresponding to the balanced or unbalanced case, $\nio{\ovl{p}(v,u)}$ is bounded by
 \[
   \noi{c_1\ufirst + c_2 \ulast + c_3 u_{\mathsmaller{-},o} + c_4 u_{\mathsmaller{+},o}} \le c_1 \noi{\ufirst} + \dots + c_4
   \noi{u_{\mathsmaller{+},o}}\,,
   \]
 the inequality follows from homogeneity and the triangle inequality of norms.
 Let $v = [a,b]$ and assume that 
 $\noi{\ufirst} \le a$, $\noi{u_{\mathsmaller{-},o}} \le a_o$, $\noi{\ulast} \le b$ 
 and $\noi{u_{\mathsmaller{+},o}} \le b_o$,
 then the right hand side is at least $c_1 a + \dots + c_4 b_o \le \ntwo{v}$ by 
 construction of the $n=1$ case. Hence the condition is satisfied.
\end{proof}

\section{Appendix: Online prediction of each $k$-th bit}
  \label{sec:moreMachines}

For all $i,k$, an $i$ modulo $k$ semimeasure is the natural generalization of even semimeasures,
where condition $i.$ holds if $|x| = i \bmod k$ and $ii.$ holds otherwise.
\begin{proposition}\label{prop:onlineKmachines}
  For all $\varepsilon > 0$, $k \ge 2$ and lower semicomputable $(i \bmod k)$-semimeasures $Q_{i
  \bmod k}$ for $i = 1\dots k$, 
  there exist $\delta > 0$, a partial computable $F$, a lower semicomputable semimeasure $P$, 
  and a sequence~$\omega$ such that for all~$n$
  \[
  \left(\prod Q_{i \bmod k}\right)(\omega_{1\dots kn}) \le (1-\delta)^n P(\omega_{1\dots kn})^{k-k\varepsilon}
  \]
  and $F(\omega_{1\dots kn}, \omega_{kn+k}) = \omega_{1\dots kn + k}$.
\end{proposition}

\begin{proof}
 The proof is analogues to the proof above.
 First consider the game for the case $n=1$. We choose 
 \[
 \delta = 2^{-\frac{1+\log k}{\varepsilon}}\,
 \]
 hence, $\delta^{1-\varepsilon} = \delta 2^{1+\log k} = 2k\delta$.
 Let $0^k$ be the string containing $k$ zeros.
 We choose $(1-\delta)P(0^k) = \delta$ and wait until $Q_{i \bmod k}(0^k) > \delta^{1-\varepsilon}$
 for some $i = 1,\dots, k$. 
 Let $i$ be the first for which this happens 
 and let $x^i$ be the leftmost string for which the $i$-th and the $k$-th bit is $1$, 
 i.e. $x^i = 0^{i-1}10^{k-i-1}1$.
 Alice's second (and last) move is $(1-\delta)P(x^i) = 1-2\delta$.
 We have that $\left(\prod Q_{i \bmod k}\right)(x^i) \le 1 - \delta^{1-\varepsilon} = 1 - 2k\delta$.
 For Bob to win, he should satisfy $\left(\prod Q_{i \bmod k}\right)(x) > (1-2\delta)^{k-\varepsilon k}$
 and this is lower bounded by $(1-2\delta)^k \ge 1 - 2k\delta$, and hence he can not win.
 In particular, $Q_{i \bmod k}(x^i) \le (1-2\delta)^{k-\varepsilon k}$. 
 %For later use we define 
 %$\beta$ be the right hand side of this inequality.

 Note that for the string $x$ on which Alice wins, we have $x_k = 1$ iff Alice had a second move.
 Knowing that Alice had a second move, we can compute $i$. Thus we can compute $x$ 
 from $k,x_k$. 
 
 For $n \ge 1$ we define $\omega$, $t_{n,i}$ inductively in a similar way as before.
 Let the event $E_{x,i}$ denote whether $Q_{i \bmod k}(x0^k) > \delta^{1-\varepsilon}$, 
 and this was not detected already for the other measures $Q_{j \bmod k}$ with $j \not= i$.
 Let $\omega_{nk+1\dots nk+k} = x^i$ if an event $E_{\omega_{1\dots kn},i}$ happens, 
 and otherwise, let $\omega_{nk+1\dots nk+k} = 0^k$. 
 Let $t_{n+1,i} = (1-2\delta)^{k-2\varepsilon k} t_{n,i}$ 
 if an event $E_{\omega_{1\dots kn},i}$ happens, 
 and $t_{n+1,j} = t_{n,j}$ for all $j \not= i$, and 
 otherwise let $t_{n+1,i} = \delta^{1-\varepsilon}t_{n,i}$.
 By induction it follows that $t_{n,i} \ge Q_{i \bmod k}(\omega_{1\dots kn})$ for all $i$ and $n$.
 For $x$ of length $kn$,  we define $t_{x,i}$ to be $t_{n,i}$ if at some point $x$ is 
 an initial segment of a candidate $\omega$ in an approximation of $\omega$ 
 as considered above. For such $x$ we define $P$ by 
 \[
 (1-\delta)^n P(x) =  \left( \prod_{i = 1\dots k} t_{x,i} \right)^{1/(k-k\varepsilon)}\;,
 \]
 Now  Proposition~\ref{prop:onlineKmachines}  follows after rescaling $\delta$.
\end{proof}

\section{Appendix: Maximal linear asymmetry}
\label{sec:2_3th}

\begin{proposition}
  There exist a sequence $\omega$, a lower semicomputable semimeasure~$P$ and 
  odd and even online lower semicomputable semimeasures 
  $\podd$ and $\pev$ exist such that for all $n$ 
  \[
    (3/2)^n\;(\modd\mev)(\omega_{1\dots 2n}) \le P(\omega_{1\dots 2n}) 
    = (\podd\pev)(\omega_2\omega_1\dots\omega_{2n}\omega_{2n-1})\,.
    \]
\end{proposition}

\begin{proof}
We consider a variant of the game defined before the proof of Proposition \ref{prop:3_4th} 
where $a,b,c,d$ are values of $2P(x)/3$ (rather than $3P(x)/4$).
Alice's winning strategy is to start with $a = c = 1/9$. 
As long as she is in winning position, she passes.
Suppose at some moment this is no longer true, thus $pr > 1/9$ and $qu > 1/9$
(see Figure~\ref{fig:onlineStrat}).
Alice's next (and last) move is $d = 4/9$ if $p \ge q$ and $b = 4/9$ otherwise. 
Note that $a + b + c + d = 6/9$ and Alice does not violate her restriction.
Let $p,q,\dots,v$ be Bob's values at the moment of Alice's last move, and let 
$p', q', \dots, v'$ denote the limits of Bob's values.
Consider the case $p \ge q$, the other case is analogous. 
We show that Bob can not win without violating his restriction, i.e. 
$q'v' > 4/9$ implies $(p'+q')(u'+v') > 1$. Indeed,
\begin{equation}\label{eq:alphaBeta}
  \left( p'+q' \right)\left( u'+v' \right) \ge \left( \sqrt{p'u'} + \sqrt{q'v'} \right)^2 .
\end{equation}
This is Cauchy's inequality $\|\vec{a}\|\cdot\|\vec{b}\| \ge |\langle \vec{a},\vec{b} \rangle |$ 
for $\vec{a} = [\sqrt{p'}, \sqrt{q'}]$ and $\vec{b} = [\sqrt{u'}, \sqrt{v'}]$. 
Because $u' \ge u$ and $p' \ge p \ge q$, the right-hand is at least 
\begin{equation}\label{eq:alphaBeta2}
  \ge \left( \sqrt{qu} + \sqrt{q'v'} \right)^2 > \left( \sqrt{\tfrac{1}{9}} + \sqrt{\tfrac{4}{9}} \right)^2 = 1.
\end{equation}

\begin{figure}
  \centering
  \begin{tikzpicture}[grow=up]

    \coordinate (odd) \treeGenerate;
    \treeLabel{odd}{$p$}{$p$}{$q$}{$q$}{$p$}{$q$}{$1$}

    \coordinate[right=3.3 of odd] (even) \treeGenerate;
    \treeLabel{even}{$r$}{$s$}{$u$}{$v$}{$1$}{$1$}{$1$}

    \node at (5.8,1) {$\Longrightarrow$};

    \coordinate[right=4.6 of even] (root) \treeGenerate;
    \treeLabel{root}{$pr$}{$ps$}{$qu$}{$qv$}{$p$}{$q$}{$1$}

      \path (root-1-2) -- node[midway,rotate=90,shift={(+0.05,0)}] {$>$} 
	+(0,1.2)  node {$\tikzfrac{1}{9}$};
      \path (root-2-2) -- node[midway,rotate=90,shift={(+0.05,0)}] {$>$} 
	+(0,1.2)  node {$\tikzfrac{1}{9}$};
      \path  (root-1-1) -- ++(0,0.7) node  {$q'v'$}
      -- node[midway,rotate=90,shift={(-0.05,0)}] {$<$} 
	++(0,0.8) node {$\tikzfrac{4}{9}$};
  \end{tikzpicture}
  \caption{Alice's winning strategy $P$ for $p \ge q$.
    \label{fig:onlineStrat}
  }
\end{figure}
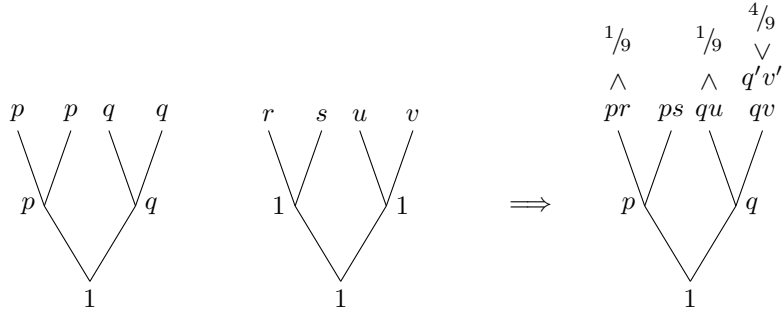

Let $t_0 = 1$.
We construct $\omega_{1\dots2n}$ together with thresholds $t_{n}$ inductively.
For $x$ of length $2n$,
let $T_{00}$ and $T_{10}$ at $x$ be the events that
$(\modd\mev)(x00) > t_{n}/9$ and $(\modd\mev)(x10) > t_{n}/9$.
Let $p_x$ and $q_x$ denote the values $\modd(x1)$ and $\modd(x0)$ at the moment we observe both 
$T_{00}$ and $T_{10}$ at $x$.
\[
\left( \omega_{2n+1}\omega_{2n+2}, t_{n+1} \right) 
= \begin{array}{rlll}
  ( 00, &  t_{n}/9 ) & \text{if $T_{00}$ does not happen at $\omega_{1\dots2n}$} \\
  ( 10, &  t_{n}/9 ) & \text{if $T_{00}$ happens at $\omega_{1\dots2n}$ but not $T_{01}$,} \\
  ( 11, &  4t_{n}/9 ) & \text{if $T_{00}$,$T_{01}$ happen at $\omega_{1\dots2n}$ and $q_x \le p_x$} \\
  ( 01, &  4t_{n}/9 ) & \text{otherwise.}
\end{array}
\]
By induction (and the game above) it follows that 
$t_n \ge (\modd\mev)(\omega_{1\dots2n})$.
We now define a lower semicomputable semimeasure~$P$ such that
\[
  P(\omega_{1\dots2n}) = (3/4)^n o_{n} e_{n}\;.
  \]
First, note that $\omega$ can be approximated as follows: start with $\omega = 00\dots$, 
if $T_{00}$ or $T_{01}$ happen, then bits $\omega_{2n}\omega_{2n+1}$ are changed accordingly,
let all subsequent bits be zero, and run the processes for $n+1$, $n+2$, etc.
For each $n$ and for each $x$ of length $2n$, 
at most one value $t_n$ can be associated to $x$. If this happens, we define
$P(x) = (6/9)^n t_n$ and $P(x) = 0$ otherwise. 
Also note that the two last cases in the definition of $\omega$ 
can not happen simultaneously, hence
$\sum \{ t_{xbb'}: b,b' \in \{0,1\} \}  \le 6t_x/9$. This implies 
$\sum \{ P(xbb') : b,b' \in \{0,1\} \} \le P(x)$ and $P$ is a semimeasure. 

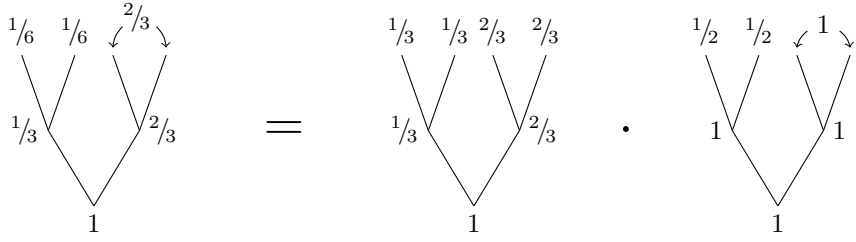
\begin{figure}
  \centering
  \begin{tikzpicture}[grow=up]
    \coordinate (root) \treeGenerate;
    \treeLabel{root}{$\tikzfrac{1}{6}$}{$\tikzfrac{1}{6}$}{}{}{$\tikzfrac{1}{3}$}{$\tikzfrac{2}{3}$}{$1$}
    \node[above=1.5cm of root-1,anchor=center] {$\!\!\tikzfrac{2}{3}\!$}
      edge[->,bend left=20] ($(root-1-1.north) + (0,0.1)$)
      edge[->,bend right=20] ($(root-1-2.north) + (0,0.1)$);

    \node at (2.5,1) {\huge{$=$}};

    \coordinate[right=5 of root] (odd) \treeGenerate;
    \treeLabel{odd}{$\tikzfrac{1}{3}$}{$\tikzfrac{1}{3}$}{$\tikzfrac{2}{3}$}{$\tikzfrac{2}{3}$}{$\tikzfrac{1}{3}$}{$\tikzfrac{2}{3}$}{$1$}

    \node  at ($(odd) + (2,1)$) {\huge{$\cdot$}};

    \coordinate[right=4 of odd] (even) \treeGenerate;
    \treeLabel{even}{$\tikzfrac{1}{2}$}{$\tikzfrac{1}{2}$}{}{}{$1$}{$1$}{$1$}
    \node[above=1.4cm of even-1,anchor=center] {$1$}
      edge[->,bend left=20] ($(even-1-1.north) + (0,0.1)$)
      edge[->,bend right=20] ($(even-1-2.north) + (0,0.1)$);
  \end{tikzpicture}
  \caption{Decomposition of $P'$ defined by Alice's strategy.
    \label{fig:decompStrat}
  }
\end{figure}

It remains to factorize $\tilde{P}(x_2x_1\dots x_{2n}x_{2n-1}) = P(x_1x_2\dots x_{2n-1}x_{2n})$ 
into two online semimeasure $\podd$ and $\pev$. The decomposition 
for $\tilde{P}$ is given in figure \ref{fig:decompStrat} for $x$ of length two 
(the two maximal cases are plotted).
This construction can be iterated, (i.e. we obtain $\podd(xbb')$ and $\pev(xbb')$ 
by multiplying the values of figure \ref{fig:decompStrat} with $P(x)$). 
In fact, $\podd$ is computable and $\pev$ is non-zero on exactly the same places 
as $\tilde{P}$.
\end{proof}

Finally, we remark that this result can be generalized for more machines 
using the generalized H\"older's inequality, 
which is in turn a generalisation of the Cauchy-Schwartz inequality:
for $r,s_1,\dots,s_k$ such that $\sum_{i = 1}^k \frac{1}{s_i} = \frac{1}{r}$,
and for vectors $\vec{u}^1, \dots, \vec{u}^k$ 
\[
\left|\left|  \vec{u}^1  \ldots  \vec{u}^k \right|\right|_r \le
\| \vec{u}^1 \|_{s_1}  \ldots \| \vec{u}^k \|_{s_k},
\]
where $\vec{a} \vec{b}$ denotes entry wise multiplication.

\section{Appendix: Chain rule for online complexity}
\label{sec:chainrule}

\begin{proposition}\label{prop:chainRule}
  $\Cev(xy) = \Cev(x) + \Cev(y|x) + O(\log |x|)$ and similar for odd complexity.
\end{proposition}

The proof is similar to the proof of symmetry of information for prefix
complexity~\cite{ZvonkinLevin}.
A conditional even semimeasure $\pev(x|y)$ is defined in the natural way, i.e. 
a function such that $\pev(.|y)$ is an even semimeasure for all $y$. 
Note that if even complexity was defined over discrete sets 
(rather than $\{0,1\}$),
the conditional variants are simply the cases where the condition is joined with 
the first symbol of the string. Hence the general version of the coding theorem 
in~\cite{onlineComplexity} implies $-\log \mev(x|y) = \Cev(x|y) + O(\log \Cev(x|y))$. 
Therefore, Proposition~\ref{prop:chainRule} follows from

\begin{lemma}\label{lem:chainRule}
 $\mev(x|\,|x|)\mev(y|x, k_x) = \Theta\left( \mev(xy|\,|x|) \right)$
 with $k_x = \lfloor -\log \mev(x| \,|x|) \rfloor$.
\end{lemma}

The proof roughly follows the proof of symmetry of information 
for prefix complexity, where it is shown that
$
 m(x) m(y|x,\lfloor -\log m(x) \rfloor) = \Theta(m(x,y))
$.
  
\begin{proof} 
  Let $m = |x|$.
 We show that the left hand side exceeds the right within a constant factor. For this it suffices to construct
 a lower semicomputable even semimeasure $\pev$ such that for all $k$ with  $2^{-k} \ge \mev(x|m)$ we have
 \[
 \pev(y|x,k) = \mev(xy|\,|x|)/2^{-k} \,.
 \]
 Indeed, assume $\mev$ is approximated from below such that at each stage $\mev$ is an even
 semimeasure. 
 At each stage, take the above function as a definition of~$\pev$, 
 as soon as $\mev(x|\,|x|) > 2^{-k}$ do not increase $\pev$ anymore. 
 Clearly $\pev$ is lower semicomputable, the ``freezing'' guarantees 
 that $\pev(\varepsilon|x,k) \le 1$ for all $x$ and $k$ and hence, it is a conditional even semimeasure.

 For the other inequality we construct an even lower semicomputable semimeasure 
 $\pev$ such that if $|z| \ge m$ then
 \[
 \pev(z|m) \ge \tfrac{1}{4} \mev(z_{1\dots m} |m) \mev(z_{m+1\dots |z|} | z_{1\dots m}, k)\,,
 \]
 for $k = \lceil -\log \mev(z_{1\dots m}|m)\rceil$. 
 Our construction of $\pev$ is as follows:
 if $|z| < m$, then $\pev(z|m) = \mev(z|m)$ and otherwise
 \[
 \sum_k \left\{ 2^{-k-1}\mev(z_{m+1\dots |x|} | z_{1\dots m}, k) 
                    : 2^{-k} \le \mev(z_{1\dots m} |m) \right\} \,.
 \]
 Note that for $|z|=m$ we have 
 $\pev(z|m) \le \sum \{2^{-k-1}: 2^{-k} \le\mev(z_{1\dots m} |m)\} \le \mev(z|m)$ 
 hence $\pev$ is an even semimeasure. 
 Moreover, $\pev$ is lower semicomputable and satisfies the condition.
\end{proof}

In~\cite{BauwensPhd},  Lemma~\ref{lem:chainRule} is combined with  P\'eter G\'acs' theorem 
that $\max\{C(C(x)|x): |x|=n\} \ge \log n - O(\log \log n)$ to obtain a more involved proof of
 Theorem~\ref{th:onlineMain} in weaker form, i.e. with a smaller (and machine dependent) linear coefficient.
\end{document}